\documentclass[pdftex,a4paper,12pt]{article}
\usepackage{amssymb,amsmath,amsfonts,nomencl,mathrsfs}
\usepackage{amsthm}  % anstelle von ntheorem
\usepackage{stmaryrd} % for \sslash
\usepackage[arrow, matrix, curve]{xy}
\usepackage[latin1]{inputenc}
\newcommand{\IDD}{\mathscr{D}}

\newcommand{\IC}{\mathbb{C}}
\newcommand{\IR}{\mathbb{R}}
\newcommand{\IMM}{\mathscr{M}}

\newcommand{\ILL}{\mathscr{L}}

\newcommand{\IHH}{\mathscr{H}}

\newcommand{\IPP}{\mathscr{P}}
\newcommand{\IFF}{\mathscr{F}}

\newcommand{\IN}{\mathbb{N}}

\newcommand*{\longhookrightarrow}%
               {\ensuremath{\lhook\joinrel\relbar\joinrel\rightarrow}}
\newcommand{\pa}{\sslash} % needs package stmaryrd

\newcommand{\Id}{{\rm d}}

\newcommand{\f}{\frac}
\newcommand{\nn}{\nonumber}

\theoremstyle{plain}            % body italics
\newtheorem{theorem}{theorem}[section]
\newtheorem{Lemma}[theorem]{Lemma}
\newtheorem{Corollary}[theorem]{Corollary}
\newtheorem{Theorem}[theorem]{Theorem}
\newtheorem{Proposition}[theorem]{Proposition}

\theoremstyle{definition}       % body roman
\newtheorem{Definition}[theorem]{Definition}

\newtheorem{Remark}[theorem]{Remark}

\newcommand{\qedpart}[2]{\hfill $\frac {#1}{#2}\blacksquare$}
\DeclareMathOperator{\End}   {End}

\newcommand{\iprod}[3][{}]{\langle{#2},{#3}\rangle_{#1}}  % inner product
\newcommand{\norm}[2][{}]{\|{#2}\|_{{#1}}}    % norm
\newcommand{\Bignorm}[2][{}]{\Bigl\|{#2}\Bigr\|_{#1}}     % norm
\renewcommand{\Re}     {\mathrm {Re}\,}
\newcommand{\e}{\mathrm e}

\newcounter{myenumi}
\newenvironment{myenumerate}[1]{
\begin{list}{(\themyenumi) }%{(\indent\themyenumi) }
  {\renewcommand{\themyenumi}{#1{myenumi}}
    \usecounter{myenumi}
    \setlength{\topsep}{0em}
    \setlength{\itemsep}{0em}
    \setlength{\leftmargin}{0em}
    \setlength{\labelwidth}{0em}
    \setlength{\labelsep}{0em}}
  }
  {
  \end{list}
  }
\newcommand{\itemref}[1]{\eqref{#1}}

\makeatletter% --> De-TeX-FAQ
\def\@fnsymbol#1{\ifcase#1\or a\or b\or c\or
   d\or e\or f\or g\or h\or
    i\else\@ctrerr\fi}
\makeatother% --> \makeatletter
\begin{document}

 \begin{titlepage}
   \renewcommand{\thefootnote}{\alph{footnote}} 

   \title{Path integrals and the essential self-adjointness of
     differential operators on noncompact manifolds}

   \author{Batu G\"uneysu\footnote{E-Mail: 
            \texttt{gueneysu@math.hu-berlin.de}}\\
     Institut f\"ur Mathematik, Humboldt-Universit\"at zu Berlin\\
     Olaf Post\footnote{Email: \texttt{olaf.post@durham.ac.uk}} %
     \footnote{\emph{On leave from:} Department of Mathematical
       Sciences, Durham
       University, England, UK}\\
     School of Mathematics, Cardiff University, Wales, UK}
\end{titlepage}

\maketitle 
\begin{abstract}
  We consider Schr\"odinger operators on possibly noncompact
  Riemannian manifolds, acting on sections in vector bundles, with
  locally square integrable potentials whose negative part is in the
  underlying Kato class. Using path integral methods, we prove that
  under geodesic completeness these differential operators are
  essentially self-adjoint on $\mathsf{C}^{\infty}_0$, and that the
  corresponding operator closures are semibounded from below. These
  results apply to nonrelativistic Pauli-Dirac operators that describe
  the energy of Hydrogen type atoms on Riemannian $3$-manifolds.
\end{abstract}

%----------------------------------------------------------------------
%
\section{Introduction}
\label{ab1}
%
%----------------------------------------------------------------------

A classical result from B.~Simon's seminal paper~\cite{Si0} states
that a Schr\"odinger operator of the form $-\Delta+V$ in the Euclidean
space $\IR^m$, with $V\colon \IR^m \to \IR$ a locally square
integrable potential, is essentially self-adjoint on
$\mathsf{C}^{\infty}_0(\IR^m)$, if the negative part of $V$ is in the
Kato class $\mathcal{K}(\IR^m)$. Note here that this fact is closely
related to quantum physics, in the sense that the Coulomb potential
$V(x)=-1/|x|$ is in the above class. Having in mind that all of the
above data can be defined on any Riemannian manifold, we are
interested in the following question in this paper: 
\begin{quote}
  \emph{To what extent can Simon\rq{}s result be extended to
    Schr\"odinger type operators acting on sections in vector bundles
    over possibly noncompact Riemannian manifolds?}
\end{quote}
Apart from a pure academic interest, this question is also
particularly motivated by the observation that it is possible to
model~\cite{G6, enciso} nonrelativistic atomic Hamiltonians on any
nonparabolic Riemannian $3$-manifold (which have to be
$\mathrm{spin}^{\IC}$, if the electron\rq{}s spin is taken into
account; in particular, the vector-valued case becomes particularly
interesting from this point of view, see Section~\ref{wasser} below).
This abstraction is desirable from the physics point of view, since
one would like to understand deeply which properties of the Euclidean
space $\IR^3$ actually guarantee certain spectral properties of
quantum systems, or other important results such as the stability of
matter~\cite{lieb}. In these situations, the corresponding potential
terms are always locally square integrable, and with some control on
the underlying Riemannian structure, their negative parts are in the
underlying Kato class, so that we basically are in the initial
situation.

Before we can formulate our main result, we have to introduce some
notation:

Let $M$ denote a smooth connected Riemannian manifold without
boundary. The geodesic distance on $M$ will be written as $\Id(x,y)$,
and $\mathrm{K}_r(x)$ will stand for the open geodesic ball with
radius $r$ around $x$, and
\[
   (0,\infty)\times M\times M
   \longrightarrow (0,\infty),\>\>(t,x,y)\longmapsto p(t,x,y)
\]
will stand for the minimal positive heat kernel on $M$.

If $F\rightarrow M$ is a smooth Hermitian vector bundle, then, abusing
the notation in the usual way, $\left|\bullet\right|_x$ stands for
the norm and the operator norm corresponding to $(\bullet,\bullet)_x$
on each (finite-dimensional) fiber $F_x$, and the scalar product and
norm corresponding to the Hilbert space $\Gamma_{\mathsf{L}^2}(M,F)$
will be written as $\left\langle \bullet,\bullet \right\rangle$ and
$\left\|\bullet\right\|$, respectively, that is,
\begin{align}
  \label{adj}
  \left\langle f_1,f_2 \right\rangle
  =\int_M (f_1(x),f_2(x))_x \mathrm{vol}(\Id x),\>\> 
  \left\|f\right\|^2=\int_M \left|f(x)\right|^2_x \mathrm{vol}(\Id x).
\end{align} 
If
$\tilde F \to M$ is a second bundle as above and if
\[
  P \colon \Gamma_{\mathsf{C}^{\infty}_0}(M,F)
  \longrightarrow \Gamma_{\mathsf{C}^{\infty}_0}(M,\tilde{F}) 
\]
is a linear differential operator, then we denote with $P^{\dagger}$
the formal adjoint of $P$ with respect to~\eqref{adj}. In particular,
the Laplace-Beltrami operator on $M$ is given in this sense as
$-\Delta=\Id^{\dagger}\Id$. The symbol $\nabla^{\mathrm{T} M}$ will
denote the Levi-Civita connection, and if nothing else is said, the
(co-)tangent bundle of $M$ will be equipped with the Hermitian
structure corresponding to the underlying Riemannian metric of $M$.
These data will be implicitely complexified, whenever necessary.

Let $E\to M$ be a smooth Hermitian vector bundle, let $\nabla$ be a
Hermitian covariant derivative in $E$ and let $V \colon
M \to \End(E)$ be a potential, that is, $V$ is a measurable
section in $\End(E)$ such that $V(x) \colon E_x\to E_x$ is
self-adjoint for almost every (a.e.) $x\in M$. Furthermore, let
$\mathcal{K}(M)$ denote the class of Kato functions\footnote{see
  Section~\ref{hau} for the definition of $\mathcal{K}(M)$ and for
  criteria for functions to be in $\mathcal{K}(M)$} on $M$. Our main
result reads as follows:

%----------------------------------------------------------------------
\begin{Theorem}
  \label{dsk} 
  Let $M$ be geodesically complete, let
  $\left|V\right|\in\mathrm{L}^2_{\mathrm{loc}}(M)$ and assume that
  $V$ admits a decomposition $V=V_1-V_2$ into potentials $V_j\geq 0$
  with $\left|V_2\right|\in\mathcal{K}(M)$. Then the operator
  $\nabla^{\dagger}\nabla/2+V$ is essentially self-adjoint on
  $\Gamma_{\mathsf{C}^{\infty}_0}(M,E)$ and its closure is semibounded
  from below.
\end{Theorem}
%----------------------------------------------------------------------

Note that the decomposition $V=V_1-V_2$ into nonnegative potentials
need not be the canonic one given by $V=V^+-V^-$, which can be defined
through the fiberwise spectral calculus of $E$.

Before we explain the strategy of the proof of Theorem~\ref{dsk}, some
remarks are in order:
%----------------------------------------------------------------------
\begin{Remark}
  \label{beme}
  \begin{myenumerate}{\alph}
  
  \item Theorem~\ref{dsk} is disjoint from the various results on
    essential self-adjointness for operators of the form
    $\nabla^{\dagger}\nabla/2+V$ that have been obtained in~\cite{Br}.
    The point here is that, in general, Kato potentials need not
    satisfy the inequality~(2.2) from~\cite{Br}, i.e., for every
    compact $K\subset M$ there are numbers $0<a_K<1$, $b_K>0$ such
    that
    \begin{align}
      \Bigl(\int_K \left|V_2(x)\right|^2_x \left|u(x)\right|^2
         \mathrm{vol}(\Id x)
      \Bigr)^{1/2}
       \leq  a_K\left\|\Delta u\right\|+b_K \left\|u\right\|
       \label{brav}
    \end{align}
    for any $u\in \mathsf{C}^{\infty}_0(M)$. However, it should be noted
    that the main strength of the results of~\cite{Br} lies in the
    fact that the authors have considered \emph{arbitrary} first order
    elliptic differential operators instead of $\nabla$. It would
    certainly be an interesting problem to see to what extent our
    probabilistic techniques below can be extended to cover the latter
    situation, which has first been considered in~\cite{lesch}.

  \item Of course, taking $E=M\times \IC$ and
    $\nabla=\Id+\mathrm{i}\beta$ with $\beta\in\Omega^1_{\IR}(M)$, we
    can deal with \emph{smooth} magnetic potentials within our
    framework. In this scalar situation, the analogue of
    Theorem~\ref{dsk} can be easily deduced from (a slight variation
    of) Theorem 1 in~\cite{grumm}, where the authors can even allow
    magnetic potentials with possibly strong local singularities. We
    refer the reader to~\cite{hinz} for the scalar situation in
    Euclidean space.
  \end{myenumerate}
\end{Remark}
%----------------------------------------------------------------------

Let us now explain the strategy (which is partially motivated
by~\cite{Si0} and~\cite{grumm}) of the proof of Theorem~\ref{dsk},
which is given in full detail in the following Section~\ref{hau}. To
this end, we assume for the rest of this section that $V$ is as in
Theorem~\ref{dsk}. Then by the main result of~\cite{G1}, it is always
possible to define the form sum $H_V$ corresponding to the Friedrichs
realization of $\nabla^{\dagger}\nabla/2$ and $V$ \emph{without any
  additional assumptions on $M$} (see Theorem~\ref{T1} below). The
main advantage of this observation is that, unlike in usual
essential-self-adjointness proofs, instead of directly proving that
$\nabla^{\dagger}\nabla/2+V$ is essentially self-adjoint on
$\Gamma_{\mathsf{C}^{\infty}_0}(M,E)$, we will prove that the latter
space is an operator core for $H_V$ (this is the content of
Theorem~\ref{T3}; Theorem~\ref{dsk} itself follows directly from the
latter result, which is summarized in Corollary~\ref{ende}). In
particular, we will use the full spectral calculus given by $H_V$.

Having said this, the first step in the proof of this operator core
property will be to deduce the following smoothing property (see
Proposition~\ref{T4} below):
\begin{align}
  \text{For any $t>0$ one has }\>\> 
  \e^{-t H_V}\Big[\Gamma_{\mathsf{L}^{2}}(M,E)\Big]
  \subset \Gamma_{\mathsf{L}^{\infty}_{\mathrm{loc}}}(M,E). \label{sdb}
\end{align} 
This result will be derived from the path integral formula
\begin{align}
  \e^{-t H_V }f(x)= \mathbb{E}\left[1_{\{t<\zeta(x)\}}
    \mathscr{V}^{x}_t \pa_t^{x,-1} f(B_t(x))\right],
  \label{ghgh}
\end{align}
where $B(x)$ is a Brownian motion starting in $x$ with lifetime
$\zeta(x)$, where
\[
\pa^x_t:E_x\longrightarrow  E_{B_t(x)} 
\]
is the corresponding stochastic parallel transport with respect to
$\nabla$, $\pa^{x,-1}_t=\pa^{x,*}_t$ its inverse, and where
\[
\mathscr{V}^{x}_t:E_x\longrightarrow  E_x
\]
is the path ordered exponential\footnote{Here, $t \Delta_k =\{0\leq
  s_1\leq \dots\leq s_k\leq t\}\subset\IR^k$ denotes the $t$-scaled
  $k$-simplex for any $k\in\IN$, $t\geq 0$.}
\begin{align}
  &\mathscr{V}^{x}_t-\mathbf{1}\nn\\
  &=\sum^{\infty}_{k=1}(-1)^k&\int_{t \Delta_k} \pa^{x,-1}_{s_1}
  V(B_{s_1}(x))\pa^{x}_{s_1}\cdots \pa^{x,-1}_{s_k}
  V(B_{s_k}(x))\pa^{x}_{s_k} \Id s_1\dots \Id s_k \label{pui}
%\dots\leq s_k\leq t} F_V(s_1)\cdots F_V(s_k) \Id s_1\dots \Id s_k,
\end{align}
(details on these processes and on formula~\eqref{ghgh}, which is one
of the main results of~\cite{G2}, are included in the following
section). Again,~\eqref{sdb} and~\eqref{ghgh} are valid \emph{without
  any additional assumptions on $M$.}

%----------------------------------------------------------------------
\begin{Remark} 
  Note that it is not possible to deduce~\eqref{sdb} directly by
  Sobolev embedding theorems for $\dim M >3$, which is the main
  motivation for the introduction of path integral techniques in this
  context.
\end{Remark}
%----------------------------------------------------------------------

In a next step, we will use finite speed propagation methods to deduce
the following result:
\begin{align}
  & \text{The set}\>\>\mathsf{D}(H_V)\cap \left.
    \Big\{f\right|\text{$f$ has a compact support}\Big\}\nn\\
  &\text{ is an operator core for $H_V$, if $M$ is geodesically
    complete.} \label{che}
\end{align}

To be precise, we will actually prove a Davies-Gaffney inequality (see
Proposition~\ref{fps}) for approximations of $H_V$ and use the fact
that this inequality always implies (is in fact equivalent) to finite
speed of propagation by the results of~\cite{coul}. Then one can use a
variant of Chernoff\rq s theorem (see Lemma~\ref{chernoff1}) to
deduce~\eqref{che}. The fact that we use finite speed propagation
methods in this context has been particularly motivated by the scalar
situation that has been considered in~\cite{grumm}, where the authors
apply this method in a similar way. As has been noted in~\cite{grumm},
this technique avoids the usage of second order cut-off functions,
which do not seem to be available without additional control on the
underlying Riemannian structure.

As a next step one can combine~\eqref{che} with~\eqref{sdb} to deduce
the following fact:
\begin{align}
  & \text{The set}\>\>\mathsf{D}(H_V)\cap 
        \Gamma_{\mathsf{L}^{\infty}_{\mathrm{loc}}}(M,E)\cap \left.
  \Big\{f\right|\text{$f$ has a compact support}\Big\}
  \label{sd1}\\
  \nonumber
  &\text{ is an operator core for $H_V$, if $M$ is geodesically complete.}
\end{align}

Then, we shall use the self-adjointness of $H_V$ to deduce that the
elements $f$ of the set~\eqref{sd1} satisfy $\nabla^{\dagger}\nabla
f\in \Gamma_{\mathsf{L}^{2}}(M,E)$.  Finally, if $M$ is geodesically
complete we can use a (local) result on Friedrichs mollifiers to prove
that $\Gamma_{\mathsf{C}^{\infty}_0}(M,E)$ is an operator core for
$H_V$, by showing that $\Gamma_{\mathsf{C}^{\infty}_0}(M,E)$ is dense
in~\eqref{sd1} with respect to the graph norm corresponding to $H_V$.

This paper is organized as follows: In Section~\ref{bew}, we first
recall some facts about Kato potentials. The rest of Section~\ref{bew}
is completely devoted to the proof of Theorem~\ref{dsk}. In
Section~\ref{wasser}, we apply Theorem~\ref{dsk} in the context of
Hydrogen type problems on Riemannian $3$-manifolds, which was
originally the main motivation for this paper. It seems as if this result has not been stated yet in this form in the literature even for the Euclidean $\IR^3$ (though it should be known in this case). Finally, in the
appendix, we have included a fact about Friedrichs mollifiers, an
abstract variant of Chernoff\rq{}s finite speed of propagation theorem
on vector bundles, and some facts about path ordered exponentials that
we will need in our probabilistic considerations.

%----------------------------------------------------------------------
%
\section{ Kato potentials and the proof of
  Theorem~\ref{dsk}}\label{bew}

%
%----------------------------------------------------------------------

Let us first clarify that in this section,
\begin{quote}
  \emph{$M$ will always be a smooth connected Riemannian manifold
    without boundary, $E\to M$ a smooth Hermitian vector bundle,
    $\nabla$ a Hermitian covariant derivative in $E$, and $V \colon
    M\to\End(E)$ a potential.}
\end{quote}
By the usual abuse of notation, we will denote the quadratic form
corresponding to a symmetric sequilinear form in some Hilbert space
with the same symbol. The symbol $H_0$ stands for the Friedrichs
realization of $\nabla^{\dagger}\nabla/2$, that is, $H_0$ is the
nonnegative self-adjoint operator in $\Gamma_{\mathsf{L}^2}(M,E)$
which corresponds to the closure $q_{H_0}$ of the quadratic form given
by the symmetric nonnegative operator $\nabla^{\dagger}\nabla/2$,
defined initially on $\Gamma_{\mathsf{C}^{\infty}_0}(M,E)$. Note the
well-known:

%----------------------------------------------------------------------
\begin{Remark}
  If $M$ is geodesically complete, then one has
  \begin{align}
    &\mathsf{D}(q_{H_0})=\left.\Big\{f\right|f\in\Gamma_{\mathsf{L}^{2}}(M,E),
\nabla f\in \Gamma_{\mathsf{L}^{2}}(M,E\otimes
\mathrm{T}^*M)\Big\},\label{f1}\\
    &q_{H_0}(f,h)= \f{1}{2}\int_M\left( \nabla f(x),\nabla
h(x)\right)_x\mathrm{vol}(\Id x),\nn
  \end{align}
  and $\Gamma_{\mathsf{C}^{\infty}_0}(M,E)$ is an operator core for
  $H_0$, and one has
  \begin{align}
    \mathsf{D}(H_0)
    =\left.\Big\{f\right|f,\nabla^{\dagger}\nabla
f\in\Gamma_{\mathsf{L}^{2}}(M,E)\Big\},\>H_0f
    =\f{1}{2}\nabla^{\dagger}\nabla f.
  \end{align}
\end{Remark}
% ----------------------------------------------------------------------

Next, we remark that $V$ defines a quadratic form in
$\Gamma_{\mathsf{L}^2}(M,E)$ by setting
\begin{align}
  &\mathsf{D}(q_{V})
  =\left.\Big\{f\right|f\in \Gamma_{\mathsf{L}^2}(M,E),\> \left( Vf,f\right)\in
\mathsf{L}^1(M)\Big\},\nn\\
  &q_{V}(f)= \int_M \left( V(x)f(x),f(x)\right)_x\mathrm{vol}(\Id x).
\end{align}

We will often require a global Kato assumption on some negative part
of $V$. Before recalling some facts on Kato functions, let us first
introduce some notation: Let $\IMM:=(\Omega,\IFF,\IFF_*,\mathbb{P})$
be a filtered probability space which satisfies the usual assumptions.
We assume that $\IMM$ is chosen in a way such that $\IMM$ carries an
appropriate family of Brownian motions
\[
B(x)\colon [0,\zeta(x))\times \Omega \longrightarrow M,\>x\in M,
\]
where $\zeta(x):\Omega\to [0,\infty]$ is the lifetime of $B(x)$. We
will freely use the fact
\begin{align}
  \mathbb{P}\{B_t(x)\in N, t<\zeta(x)\}=\int_N p(t,x,y)\mathrm{vol}(\Id y)\>\>\text{
for any measurable $N\subset M$}\nn
\end{align}
in the following.

Now a measurable function $w:M\to\IC$ is said to be in the \emph{Kato
  class} $\mathcal{K}(M)$ of $M$, if
\begin{align}
&\lim_{t\to 0+}\sup_{x\in M} \mathbb{E}\left[\int^t_0 1_{\{s<\zeta(x)\}}
\left|w(B_s(x))\right|\Id s \right]= 0,\>\>\text{ which is equivalent to} \label{ka}\\
&\lim_{t\to 0+}\sup_{x\in M} \int^t_0 \int_M p(s,x,y) \left|w(y)\right|
\mathrm{vol}(\Id y) \Id s= 0.
\end{align}
The \emph{local Kato class} $\mathcal{K}_{\mathrm{loc}}(M)$ is defined
in the obvious way,
\[
\mathcal{K}_{\mathrm{loc}}(M):=\left.\Big\{w\right|1_Kw \in\mathcal{K}(M)\>\text{
for all compact $K\subset M$} \Big\}\supset \mathcal{K}(M),
\]
and generally, $\mathcal{K}_{\mathrm{loc}}(M)$ may depend on the
Riemannian structure of $M$.

For future reference, we note:
%----------------------------------------------------------------------
\begin{Lemma}
  \label{dsaa}
  \begin{myenumerate}{\alph}
  \item
    \label{dsaa.a}
    One has $\mathcal{K}(M)\subset \mathsf{L}^1_{\mathrm{loc}}(M)$ and
    $\mathsf{L}^{\infty}(M)\subset \mathcal{K}(M)$.

  \item
    \label{dsaa.b}
    For any $w\in\mathsf{L}^1_{\mathrm{loc}}(M)$ and a.e.  $x\in M$
    one has
    \begin{align}
      \mathbb{P}\Big\{w(B_{\bullet}(x))\in 
       \mathsf{L}^1_{\mathrm{loc}}[0,\zeta(x))\Big\}=1.
      \label{qtm0}
    \end{align}

  \item
    \label{dsaa.c}
    For any $w\in \mathcal{K}_{\mathrm{loc}}(M)$ and all $x\in M$ one
    has
    \begin{align}
      \mathbb{P}\Big\{w(B_{\bullet}(x))\in
        \mathsf{L}^1_{\mathrm{loc}}[0,\zeta(x))\Big\}=1.
      \nn
    \end{align}

    \item
    \label{dsaa.d}
    For any $w \in\mathcal{K}(M)$, $t\geq 0$, one has
    \begin{align}
      \sup_{x\in M}\mathbb{E}\left[1_{\{t<\zeta(x)\}}
        \e^{\int^t_0 |w(B_s(x))|\Id
          s}\right]<\infty.
      \label{qtm}
    \end{align}
  \end{myenumerate}
\end{Lemma}
%----------------------------------------------------------------------
\begin{proof} Part~\itemref{dsaa.a} is an elementary result which is
  included in~\cite{G1}, and the
  parts~\itemref{dsaa.b}--\itemref{dsaa.d} are included in Prop.~2.4
  and Prop.~2.5 in~\cite{G2}.
\end{proof}
%----------------------------------------------------------------------
Let us now point out that \cite{G1} that one always has 
\[
\mathsf{L}^{\infty}(M)\subset \mathcal{K}(M) \subset \mathsf{L}^{1}_{\mathrm{loc}}(M),
\]
but with some control on the Riemannian
structure of $M$, one can easily produce a large class of Kato
functions. To this end, we first note the following highly nontrivial
self-improvement result of on-diagonal upper estimates for $p(t,x,y)$,
which will be very useful in the following:

%----------------------------------------------------------------------
\begin{Theorem}
  \label{si}
  Assume that there is a $C>0$ and a $t_0\in (0,\infty]$ such that
  \begin{equation*}
    \sup_{x\in M} p(t,x,x)
    \leq \f{C}{t^{\dim M/2}}\>\>\text{ for all $0<t\leq t_0$.}
  \end{equation*}
  Then there are $C_1$, $C_2>0$ such that
  \begin{equation*}
    \sup_{x,y\in M} p(t,x,y)
    \leq \f{C_1}{t^{ \dim M/2}}\e^{-\Id(x,y)^2/(C_2 t)}\>\>
    \text{ for all $0<t\leq t_0$.} 
  \end{equation*}
\end{Theorem}
%----------------------------------------------------------------------

The reader may find a proof of this result in~\cite{grq} (see
Theorem~1.1 therein for a more general result). \\
For any $p\geq 1$ let $\mathsf{L}^p_{\mathrm{u,loc}}(M)$ denote the
space of uniformly locally $p$-integrable functions on $M$, that is, a
measurable function $v \colon M\to\IC$ is in
$\mathsf{L}^p_{\mathrm{u,loc}}(M)$, if and only if
\begin{align}
\sup_{x\in
M} \int_{\mathrm{K}_1(x)}\left|v(y)\right|^p\mathrm{vol}(\Id
y)<\infty. 
\end{align}
Note the simple inclusions
\[
\mathsf{L}^p(M)\subset \mathsf{L}^p_{\mathrm{u,loc}}(M)\subset \mathsf{L}^p_{\mathrm{loc}}(M).
\]
Now one has the following result:

%----------------------------------------------------------------------
\begin{Proposition}\label{dhj} Let $p$ be such that $p\geq 1$ if $m=1$, and $p>m/2$ if $m\geq 2$.\\
(a) If there is $C>0$ and a $t_0>0$ such that
\begin{equation}
    \sup_{x\in M} p(t,x,x)
    \leq \f{C}{t^{\dim M/2}}\>\>\text{ for all $0<t\leq t_0$,}\label{heatbound}
  \end{equation}
then one has 
\begin{align}
\mathsf{L}^{p}(M)+\mathsf{L}^{\infty}(M)\subset \mathcal{K}(M).\label{incl}
\end{align}
(b) Let $M$ be geodesically complete, and assume that there are constants $C_1,\dots ,C_6, t_0 >0$ such that for all $0<t \leq t_0$, $x,y\in M$, $r>0$ one has 
$$
\mathrm{vol}(\mathrm{K}_r(x))\leq C_1 r^{\dim M} \mathrm{e}^{C_2 r} 
$$
and
\begin{align}
\f{C_3}{t^{ \dim M/2}} \mathrm{e}^{-  C_4 \Id(x,y)^2/t}\leq p(t,x,y)
\leq \f{C_5}{t^{ \dim M/2}}\mathrm{e}^{-  C_6\Id(x,y)^2/t}.\nn
\end{align}
Then one has
\begin{align}
\mathsf{L}^{p}_{\mathrm{u,loc}}(M)+\mathsf{L}^{\infty}(M)\subset \mathcal{K}(M).\label{incl2}
\end{align}
\end{Proposition}
%----------------------------------------------------------------------
\begin{proof} a) Indeed, Theorem~\ref{si} implies the existence of a $\tilde{C}>0$
  such that for all $0<t \leq t_0$ one has\footnote{Of course this
    inequality can also be deduced with an elementary argument.}
  \begin{equation*}
    \sup_{x,y\in M} p(t,x,y)\leq \f{\tilde{C}}{t^{ \dim M/2 }}.
  \end{equation*}
  Now we can directly apply Proposition~2.8 in~\cite{G1} (the
  corresponding proof is elementary and essentially only uses
  H\"older\rq{}s inequality).\\
b) We can use Theorem 3.3 from \cite{kt} with $\nu:=m$, $\beta:=2$, $V(r):=C_1 r^m \mathrm{e}^{C_2 r}$, $\Phi_1(s):=C_3\mathrm{e}^{-C_4 s^2}$, $\Phi_2(s):=C_5\mathrm{e}^{-C_6 s^2}$ to deduce the asserted inclusion (keeping $\mathsf{L}^{\infty}(M)\subset \mathcal{K}(M)$ in mind). Indeed, one just has to note that
\begin{align}
\int^{\infty}_1\f{\max(r^m \mathrm{e}^{C_2 r},r^m)\mathrm{e}^{-C_6 r^2}}{r}\Id r=\int^{\infty}_1 \mathrm{e}^{C_2 r}r^{m-1}\mathrm{e}^{-C_6 r^2}\Id r<\infty,
\end{align}
which is obvious.
\end{proof}

%----------------------------------------------------------------------
\begin{Remark}
  Let us note that (\ref{heatbound}) is satisfied, for example, if $M$ is geodesically
  complete with Ricci curvature bounded from below and a positive
  injectivity radius (see example~\cite{kt}, p.~110). The reader may find these and several other aspects on Kato functions in \cite{G1} and, particularly, in \cite{kt}.
\end{Remark}
%----------------------------------------------------------------------

The following result is also included in~\cite{G1}. It shows that,
remarkably, one can always define the form sum of $H_0$ and $V$ under
the following very weak assumptions on $V$:

%----------------------------------------------------------------------
\begin{Theorem}
  \label{T1}
  Let $V$ be such that there is a decomposition $V=V_1-V_2$ into
  potentials $V_j\geq 0$ with
  $\left|V_1\right|\in\mathsf{L}^1_{\mathrm{loc}}(M)$ and
  $\>\left|V_2\right|\in\mathcal{K}(M)$. Then one has
  \begin{equation}
    \mathsf{D}(q_{H_0}+q_{V})
    =\mathsf{D}(q_{H_0})\cap \mathsf{D}(q_{V_1}),
    \label{doy}
  \end{equation}
  and $q_{H_0}+q_{V}$ is a densely defined, closed and semibounded
  from below quadratic form in $\Gamma_{\mathsf{L}^2}(M,E)$.
\end{Theorem}
%----------------------------------------------------------------------

In the situation of Theorem~\ref{T1}, the form sum $H_0\dotplus V$
will be denoted with $H_V$, that is, $H_V$ is the self-adjoint
semibounded from below operator corresponding to $q_{H_0}+q_{V}$.
 
%----------------------------------------------------------------------
\begin{Remark}
  In the situation of Theorem~\ref{T1}, assume that $M$ is
  geodesically complete. Then Proposition~2.14 in~\cite{G1} states
  that $\Gamma_{\mathsf{C}^{\infty}_0}(M,E)$ is a form core for $H_V$.
\end{Remark}
%----------------------------------------------------------------------

Let us add the following simple observation:

%----------------------------------------------------------------------
\begin{Lemma}
  \label{hil}
  Let $\left|V\right|\in\mathsf{L}^2_{\mathrm{loc}}(M)$ and assume
  that there is a decomposition $V=V_1-V_2$ into potentials $V_j\geq
  0$ with $\left|V_2\right|\in\mathcal{K}(M)$. Furthermore, let
  $\tilde H_{V,\min}$ denote the operator $\nabla^{\dagger}\nabla
  /2+V$ with domain of definition
  $\Gamma_{\mathsf{C}^{\infty}_0}(M,E)$, and let $H_{V,\min}:=
  \overline{\tilde H_{V,\min}}$. Then one has $H_{V,\min}\subset
  H_V$.
\end{Lemma}
%----------------------------------------------------------------------

\begin{proof} Since $H_V$ is closed, it is sufficient to prove $\tilde
  H_{V,\min} \subset H_V$. But if
  $f\in\Gamma_{\mathsf{C}^{\infty}_0}(M,E)$, $h \in
  \mathsf{D}(q_{H_V})$, then $f \in \mathsf{D}(q_{H_V})$ and we have
    \begin{equation}
      q_{H_V}(f,h)=\f{1}{2}\iprod{\nabla^{\dagger}\nabla f}{ h} +
      \iprod {Vf} h,\label{opo}
    \end{equation}
    so $f\in \mathsf{D}(H_V)$ and $H_Vf=\f{1}{2}\nabla^{\dagger}\nabla
    f + Vf$. 
\end{proof}
 %-------------

As we have already remarked in the introduction, an essential step in
the proof of Theorem~\ref{dsk} will be to deduce an $\mathsf{L}^2
\leadsto\mathsf{L}^{\infty}_{\mathrm{loc}}$ smoothing property of the
Schr\"odinger semigroup
\[
  (\e^{-t H_V })_{t\geq 0}\subset \ILL(\Gamma_{\mathsf{L}^{2}}(M,E)),
\]
which will be deduced from a path integral formula for $\e^{-t
  H_V}$. In order to formulate the latter formula in our geometric
context, for any $t\geq 0$ the stochastic parallel transport with
respect to $(B(x),\nabla)$ will be written as a pathwise unitary map
\[
  \pa^x_t:E_x\longrightarrow  E_{B_t(x)},\>\>
  \text{ defined in $\{t<\zeta(x)\}\subset \Omega$}. 
\]

Now Theorem~2.11 in~\cite{G2} states the following Feynman-Kac type
path integral formula:

%----------------------------------------------------------------------
\begin{Theorem}
  \label{hb3}
  In the situation of Theorem~\ref{T1}, for a.e. $x\in M$, there is a
  unique process
  \[
  \mathscr{V}^{x} \colon
  [0,\zeta(x))\times \Omega \longrightarrow  \End(E_x)
  \]
  which satisfies 
  \begin{equation}
    \f{\Id  \mathscr{V}^{x}_{t}}{\Id t}
    = -\mathscr{V}^{x}_{t}\Big( \pa^{x,-1}_t V(B_t(x)) 
    \pa^{x}_t\Big),\>\mathscr{V}^{x}_{0}=\mathbf{1}\label{gg6}
  \end{equation}
  pathwise in the weak sense, and for any $f\in
  \Gamma_{\mathsf{L}^2}(M,E)$, $t\geq 0$, a.e. $x\in M$ one has
  \begin{align}
    \e^{-t H_V }f(x)
    = \mathbb{E}\left[1_{\{t<\zeta(x)\}} \mathscr{V}^{x}_t 
        \pa_t^{x,-1} f(B_t(x))\right].
  \label{ii7}
  \end{align}
\end{Theorem}
%----------------------------------------------------------------------

%----------------------------------------------------------------------
\begin{Remark}
  The set of $x$ for which $\mathscr{V}^{x}$ exists is, by definition,
  equal to the set $x$ for which one has~\eqref{qtm0} for $w=|V|$, and
  if $x$ is in this set, then the asserted formula~\eqref{pui} from
  the introduction follows from Lemma~\ref{poe}.
\end{Remark}
%----------------------------------------------------------------------

We will use~\eqref{ii7} to deduce:

%----------------------------------------------------------------------
\begin{Proposition}
  \label{T4}
  In the situation of Theorem~\ref{T1}, one
  has
  \begin{align}
    \e^{-t H_V}\Big[\Gamma_{\mathsf{L}^{2}}(M,E)\Big]
    \subset \Gamma_{\mathsf{L}^{\infty}_{\mathrm{loc}}}(M,E)
    \text{ for any $t>0$.}
    \label{ghd}
    % \e^{-t H_V}\in\ILL\Big(\Gamma_{\mathsf{L}^{2}}(M,E),
    %\Gamma_{\mathsf{L}^{\infty}}(M,E)\Big).\label{ghd}
  \end{align} 
\end{Proposition}
%----------------------------------------------------------------------

%----------------------------------------------------------------------
\begin{Remark}
  Note that Lemma~\ref{dsaa}, Theorem~\ref{T1}, Theorem~\ref{hb3}, and
  Proposition~\ref{T4} are all valid \emph{without any further
    assumptions on the Riemannian structure of $M$}.
\end{Remark}
%----------------------------------------------------------------------

%----------------------------------------------------------------------
\begin{proof}[Proof of Proposition~\ref{T4}]
  We define scalar potentials
  $v_j:M\to [0,\infty)$, $v:M\to \IR$ by
  \begin{align}
    v_1(\bullet):= \min\sigma(V_1(\bullet)),\>v_2(\bullet):=
    \max\sigma(V_2(\bullet)),\>v(\bullet):=v_1(\bullet)-v_2(\bullet).\nn
  \end{align}
  Let $x$ be such that~\eqref{qtm0} holds for $w=|V_1|$ and $w=
  |V_2|$. Then $\mathscr{V}^{x}$ exists, and $V\geq v\mathbf{1}$,
  Lemma~\ref{poe} and $-v\leq v_2$ imply
  \begin{align}
    \left|\mathscr{V}^{x}_t\right|_x1_{\{t<\zeta(x)\}}&\leq \e^{-\int^t_0
v(B_s(x))\Id s}1_{\{t<\zeta(x)\}}\nn\\
    &\leq \e^{\int^t_0 v_2(B_s(x))\Id
      s}1_{\{t<\zeta(x)\}}\>\>\text{ $\mathbb{P}$-a.s. for any $t\geq
      0$,}\nn
  \end{align}
  so that for any $t>0$ one has
  \begin{align}
    &\left|\mathbb{E}\left[1_{\{t<\zeta(x)\}} \mathscr{V}^{x}_t \pa_t^{x,-1}
f(B_t(x))\right]\right|_x\nn\\
    \leq \> &\mathbb{E}\left[1_{\{t<\zeta(x)\}} \e^{\int^t_0 v_2(B_s(x))\Id s} 
\left| f(B_t(x))\right|_{B_t(x)}\right]\nn\\
    \leq \>&  \sqrt{\mathbb{E}\left[1_{\{t<\zeta(x)\}} \e^{2\int^t_0 v_2(B_s(x))\Id
s} \right] }  \sqrt{\mathbb{E}\left[1_{\{t<\zeta(x)\}} \left|
f(B_t(x))\right|^2_{B_t(x)}\right] }\nn\\
    = \>& \sqrt{\mathbb{E}\left[1_{\{t<\zeta(x)\}}
        \e^{2\int^t_0 v_2(B_s(x))\Id s} \right]}
    \sqrt{\int_{M} \left| f(y)\right|^2_y p(t,x,y)\mathrm{vol}(\Id
      y)}\label{fdg}.
    % \leq \>& \left\| f\right\| \sqrt{\sup_{x,y\in M}
    %   p(t,x,y)}\sqrt{\mathbb{E}\left[1_{\{t<\zeta^x\}} %\e^{2\int^t_0
%v_2(B_s(x))\Id s} \right]} <\infty,
  \end{align}
  Since for any $h\in \mathsf{L}^1(M)$, the function
  \[
  M\longrightarrow \IC, \>\>z\longmapsto \int_{M} h(y)
  p(t,z,y)\mathrm{vol}(\Id y)
  \]
  is in $\mathsf{C}^{\infty}(M)$ (see Theorem 7.19 in~\cite{gre}), we
  can use~\eqref{qtm} with $w= v_2$ to deduce that for any compact
  $K\subset M$ one has
  \[
  \sup_{z\in K}\left(\mathbb{E}\left[1_{\{t<\zeta(z)\}}
      \e^{2\int^t_0 v_2(B_s(z))\Id s} \right]\int_{M} \left|
      f(y)\right|^2_y p(t,z,y)\mathrm{vol}(\Id y)\right)<\infty,
  \]
  so that, in view of~\eqref{fdg}, the assignment
  \[
  x\longmapsto \mathbb{E}\left[1_{\{t<\zeta(x)\}} \mathscr{V}^{x}_t \pa_t^{x,-1}
f(B_t(x))\right]  
  \]
  defines an element of
  $\Gamma_{\mathsf{L}^{\infty}_{\mathrm{loc}}}(M,E)$, and~\eqref{ghd}
  is implied by the path integral formula from Theorem~\ref{hb3}.
\end{proof}
%----------------------------------------------------------------------

Next, we are going to deduce a finite propagation speed result, which
will be used later on to prove that the compactly supported elements
of $\mathsf{D}(H_V)$ are an operator core for $H_V$ under geodesic
completeness. The essential observation is that finite speed of
propagation is always implied by a Davies-Gaffney type inequality,
through a Paley-Wiener type theorem~\cite{coul}. As we have already
remarked in the introduction, we have borrowed this method
from~\cite{grumm}.

%----------------------------------------------------------------------
\begin{Proposition}
  \label{fps}
  Let $M$ be geodesically complete.
  \samepage
  \begin{myenumerate}{\alph}
  \item  
    \label{fps.a}
    If $V$ is bounded, then there is a constant $D>0$ such that for
    all open sets $U_1,U_2\subset M$, all $f_1,f_2\in
    \Gamma_{\mathsf{L}^{2}}(M,E)$ with $\mathrm{supp}(f_j)\subset U_j$
    and all $t>0$ one has
    \begin{align}
      \left| \left\langle\e^{-t H_V}f_1,f_2
        \right\rangle\right|\leq \e^{Dt} \e^{-
        \Id(U_1,U_2)^2/(4t)}\norm{f_1}\norm{f_2}.\label{ga}
    \end{align}

  \item 
    \label{fps.b}
    Let $V$ be as in Theorem~\ref{T1} and assume $H_V\geq 0$.  Then
    for any compactly supported $f\in \Gamma_{\mathsf{L}^{2}}(M,E)$
    and any $t>0$, the section $\cos(t\sqrt{H_V})f$ has a compact
    support.
  \end{myenumerate}
\end{Proposition}
%----------------------------------------------------------------------

%----------------------------------------------------------------------
\begin{proof} 
  \itemref{fps.a}~Under the assumption that $V$ is bounded and
  nonnegative, we are going to prove~\eqref{ga} with $D=0$, which of
  course proves the assertion. To this end, we are going to use the well-known exponential-weight
  method, that goes back to~\cite{gaff} (see also~\cite{coul}): Let $q
  \colon M\to\IR$ be a bounded Lipschitz function with $|\Id q|\leq C$
  a.e. in $M$. For any $f\in \Gamma_{\mathsf{C}^{\infty}_0}(M,E)$,
  Lemma~\ref{hil} and the Sobolev product rule
  \begin{align}
    \nabla (\e^{2q} \e^{-tH_V}f)
    = \Id \e^{2q}\otimes \e^{-tH_V}f +\e^{2q} 
          \nabla \e^{-tH_V}f
  \end{align}
  imply
  \begin{align}    
    &\f{\Id}{\Id t} \norm{\e^{q}  \e^{-tH_V}f}^2\nn\\
    &=-2\Re \left\langle \nabla^{\dagger}\nabla \e^{-tH_V}f,\e^{2q}
\e^{-tH_V}f\right\rangle-2 \left\langle V \e^{-tH_V}f,\e^{2q}
\e^{-tH_V}f\right\rangle\nn\\
    &=  -2\Re\left\langle \e^{q} \nabla \e^{-tH_V}f ,\e^{q}\Id q\otimes \e^{-tH_V}f
\right\rangle-2\norm{\e^{q}\nabla \e^{-tH_V}f }^2\nn\\
    &\ \ -2 \left\langle V \e^{-tH_V}f,\e^{2q} \e^{-tH_V}f\right\rangle.
  \end{align}
  Using Cauchy-Schwarz on the fibers for the first term and $V\geq 0$
  for the last term, the latter expression can be estimated by
  \begin{align}
    \leq & \ 2 \int_M\e^{q(x)} \left| \nabla \e^{-tH_V}f(x) \right|_x \e^{q(x)}
\left|\Id q(x)\right|_x \left|\e^{-tH_V}f(x) \right|_x\mathrm{vol}(\Id
x)\nn\\
    & \ -2\norm{\e^{q}\nabla \e^{-tH_V}f }^2,
  \end{align}
  which, using $XY\leq X^2+Y^2/4$, is 
  \begin{align}
    \leq \f{1}{2} \left\| \e^{q}  \left|\Id q\right| \e^{-tH_V}f\right\|^2\leq
\f{C^2}{2} \left\| \e^{q}  \e^{-tH_V}f \right\|^2.
  \end{align}
  Thus, setting $\mathscr{E}_{f,q}(t):=\norm{\e^{q}
    \e^{-tH_V}f}^2$, putting everything together and using
  Gronwall, we arrive at
  \begin{align}
    \mathscr{E}_{f,q}(t)
    \leq \e^{C^2 t/2} \mathscr{E}_{f,q}(0)\label{gron}.
  \end{align}
  Now let $U_1,U_2$ be disjoint, let $f\in
  \Gamma_{\mathsf{C}^{\infty}_0}(M,E)$ with $\mathrm{supp}(f)\subset
  U_2$, and let $a>0$. Then the function $q:=a\Id(\bullet,U_2)$ is bounded and Lipschitz with $|\Id q| \leq a $ a.e.\ in $M$
  and~\eqref{gron} implies
  \begin{align}
    &\norm{1_{U_1}\e^{-t H_V}f }^2\nn\\
    &\leq  \e^{-a\Id(U_1,U_2)  }  \e^{a^2 t/2}
\mathscr{E}_{f,q}(0)\nn\\
    &= \e^{-a\Id(U_1,U_2)}\e^{a^2 t/2} \int_{U_2} \left|
f(x) \right|^2_x \e^{a\Id(x,U_2)  }\mathrm{vol}(\Id x)\nn\\
    &= \e^{-a\Id(U_1,U_2)
    }\e^{a^2 t/2} \norm{f }^2,
  \end{align}
  so that by choosing $a$ appropriately
  \begin{align}
    \norm{1_{U_1}\e^{-t H_V}f }
    \leq   \e^{ -\Id(U_1,U_2)^2/(4t)} \norm{f},
    \label{swsw}
  \end{align}
  which carries over to $f_2$ by a density argument. Finally, we have
  \begin{align}
    &\left| \left\langle\e^{-t H_V}f_1,f_2
      \right\rangle\right|=\left| \left\langle
        f_1,1_{U_1}\e^{-t H_V}f_2 \right\rangle\right|\leq
    \e^{- \Id(U_1,U_2)^2/(4t)}\norm{f_1}\norm{f_2}
  \end{align}
  by Cauchy-Schwarz and~\eqref{swsw}, and everything is proved.
  
  \itemref{fps.b}~It is sufficient to prove that for any $U_j$, $f_j$
  as in~\itemref{fps.a} and any $0<s<\Id(U_1,U_2)$ one has
  \begin{align}
    \left\langle\cos\left(s\sqrt{H_V}\right)f_1,f_2 \right\rangle =0.\label{ga2}
  \end{align}
  Indeed, the latter implies that if $\mathrm{supp}(f)\subset
  \mathrm{K}_r(x)$ for some $r>0$, $x\in M$, then for any $t>0$ one
  has
  \begin{align}
    \mathrm{supp}\left(\cos\left(t\sqrt{H_V}\right)f\right)\subset
\overline{\mathrm{K}_{r+t}(x)}, 
  \end{align}
  and the latter set is compact by the geodesic completeness of $M$.
  It remains to prove~\eqref{ga2}.

  If $V$ is bounded and $H_V\geq 0$, then~\eqref{ga2} follows directly
  from~\itemref{fps.a}: Indeed, one can use the same arguments as
  those in the proof of theorem 3.4 in~\cite{coul} to see this.
  Essentially, one has to use a variant of the Paley-Wiener theorem,
  which has to be applied to an appropriately rescaled version of the
  analytic function $z\mapsto \left\langle\e^{-z H_V}f_1,f_2
  \right\rangle$, $\Re z>0$.

  Next, we assume that $V$ is locally integrable and bounded from
  below with $H_V\geq 0$. Then putting $V_n:=\min(V,n)$ for $n\in\IN$
  (in the sense of the fiberwise spectral calculus of $E$) we find by
  the above that~\eqref{ga2} is satisfied for $V$ replaced with $V_n$,
  but monotone convergence of quadratic forms (see the proof of
  theorem 2.11 in~\cite{G2}) gives $H_{V_n}\to H_V$ as $n\to\infty$ in
  the strong resolvent sense, which implies~\eqref{ga2}.

  Finally, if $V$ is as in Theorem~\ref{T3} and $H_V\geq 0$, let us
  set $V_n:=\max(-n,V)$. Then each $V_n$ is locally integrable and
  bounded from below with $H_{V_n}\geq 0$ and again everything follows
  from the above and monotone convergence of quadratic forms (this is
  also included in the proof of theorem 2.11 in~\cite{G2}).
\end{proof}
%----------------------------------------------------------------------

Now we are in the position to prove the main result of this paper:

%----------------------------------------------------------------------
\begin{Theorem}
  \label{T3} Let $M$ be geodesically complete, let
  $\left|V\right|\in\mathsf{L}^2_{\mathrm{loc}}(M)$ and assume that
  $V$ has a decomposition $V=V_1-V_2$ into potentials $V_j\geq 0$ with
  $\left|V_2\right|\in\mathcal{K}(M)$. Then
  $\Gamma_{\mathsf{C}^{\infty}_0}(M,E)$ is an operator core for $H_V$
  and one has
  \begin{align}
    \mathsf{D}(H_V)
    =\left.\Big\{f\right|f,(\nabla^{\dagger}\nabla+V)f
         \in\Gamma_{\mathsf{L}^{2}}(M,E)\Big\}.
    \label{dof}
  \end{align}
\end{Theorem}
%----------------------------------------------------------------------

%----------------------------------------------------------------------
\begin{proof}
  We have to prove that $\Gamma_{\mathsf{C}^{\infty}_0}(M,E)$ is dense
  in $\mathsf{D}(H_V)$ with respect to the graph norm
  $\left\|\bullet\right\|_{H_V}$. This will be proven in four steps:

%By Proposition~\ref{T4}, it is sufficient to prove that
 % $\Gamma_{\mathsf{C}^{\infty}_0}(M,E)$ is dense in
  %$\mathsf{D}_{\infty,\mathrm{loc}}(H_V)$ with respect to the graph
 % norm $\left\|\bullet\right\|_{H_V}$. This will be proven in four
%  steps:
  \begin{myenumerate}{\Roman}
  \item
    \label{t3.i}
    \emph{If $\chi \in \mathsf{C}^{\infty}_0(M)$ and $f \in
      \mathsf{D}(H_V)$, then $\chi f \in \mathsf{D}(H_V)$ and}
    \begin{equation}
      \label{eq:chi.f.dom}
      H_V (\chi f)
      = \chi H_V f -  \nabla_{(\Id \chi)^\sharp} f
      - \f{1}{2}(\Delta \chi) f.
    \end{equation}
    Here, $(\Id \chi)^\sharp$ denotes the vector field corresponding
    to the $1$-form $\Id \chi$ (with respect to the underyling
    Riemannian metric).

    \emph{Proof.} We first note that the Sobolev product rule
    \begin{align}
      \nabla (\chi f) = (\Id \chi)\otimes f + \chi \nabla
      f\label{prod}
    \end{align}
    (which is applicable in view of~\eqref{f1} and~\eqref{doy}) shows
    that $\chi f$ is in $\mathsf{D}(q_{H_V})$, so that in order to
    prove $\chi f \in \mathsf{D}(H_V)$, it is sufficient to construct
    a $u \in \Gamma_{\mathsf{L}^2}(M,E)$ such that
    \begin{equation}
      \label{eq:dom.h.v}
      q_{H_V}(\chi f, h) = \iprod u h
    \end{equation}
    for all $h \in \mathsf{D}(q_{H_V})$, where then $H_V(\chi f)$ is
    given by $u$. To this end, we calculate
    \begin{align*}
      &q_{H_V}(\chi f, h)\\
      &= \f{1}{2}\iprod{\nabla (\chi f)}{\nabla h} + \iprod {V(\chi f)} h\\
      &= \f{1}{2}\iprod{\nabla f}{\nabla(\chi h)} -
      \f{1}{2}\iprod{\nabla f} {(\Id \chi)\otimes h} +
      \f{1}{2}\iprod{(\Id \chi) \otimes f} {\nabla h}
      +  \iprod{Vf} {\chi h}\\
      &=\iprod{H_V f} {\chi h} - \iprod{\nabla_{(\Id \chi)^\sharp} f}
      h + \f{1}{2}\iprod{(\Id^{\dagger} \Id \chi) f} h,
    \end{align*}
    where we have used~\eqref{prod} in the second equality, and $f \in
    \mathsf{D}(H_V)$ together with an integration by parts formula
    (Lemma~8.8 in~\cite{Br}) and the Sobolev product rule
    \[
    \nabla^{\dagger}(\alpha\otimes f)= (\Id^{\dagger} \alpha) f
    -\nabla_{\alpha^\sharp} f
    \]
    for (sufficiently) smooth $1$-forms $\alpha$ in the third
    equality.  In particular, we found a candidate $u$
    in~\eqref{eq:dom.h.v} and it has the desired form as
    in~\eqref{eq:chi.f.dom}.  \qedpart 14

  \item
    \label{t3.ii}
    \emph{The space
      \[
      \mathsf{D}^0(H_V)
      :=\mathsf{D}(H_V)\cap
      \left.\Big\{f\right|\text{$f$ has a compact support}\Big\}
      \]
      is dense in $\mathsf{D}(H_V)$ with respect
      to $\norm[H_V] \bullet$.}
    
    \emph{Proof.} By adding a constant, we can assume that $H_V\geq
    0$. But then the result readily follows from combining
    Proposition~\ref{fps} with Lemma~\ref{chernoff1}. \qedpart 12

 \item
   \label{t3.iii}
   \emph{The space
     \[
     \mathsf{D}^0_{\infty,\mathrm{loc}}(H_V)
     :=\mathsf{D}^0(H_V)\cap \Gamma_{\mathsf{L}^{\infty}_{\mathrm{loc}}}(M,E)
     \]
     is dense in $\mathsf{D}^0(H_V)$ with respect
     to $\norm[H_V] \bullet$.}
   
   \emph{Proof.} Let $f \in\mathsf{D}^0(H_V)$ and take $r>0$, $y\in M$
   with $\mathrm{supp}(f)\subset \mathrm{K}_r(y)$. Furthermore, pick a
   $\chi\in\mathsf{C}^{\infty}_0(M)$ with $\chi=1$ in
   $\mathrm{K}_{r+1}(y)$ and set $f_t:=\chi \e^{-t H_V} f$ for
   any $t> 0$. Then Proposition~\ref{T4} implies $f_t\in
   \mathsf{D}^0_{\infty,\mathrm{loc}}(H_V)$ and clearly
   $\norm{f_t-f}\to 0$ as $t\to 0+$. Furthermore,~\itemref{t3.i}
   implies $H_V(\chi f)= H_V f$ and also
   \begin{align}
     H_V(f_t-f)
     = \ &\chi H_V  \e^{-t H_V} f- \nabla_{(\Id \chi)^\sharp}
          \e^{-t H_V} f -\f{1}{2} (\Delta\chi ) 
          \e^{-t H_V} f -\chi H_Vf\nn\\
         &+\nabla_{(\Id \chi)^\sharp} f +\f{1}{2} (\Delta\chi ) f. \nn
   \end{align}
   Now it is easily seen that $\norm{H_V(f_t-f)}\to 0$ as $t\to 0+$.
   \qedpart 34

 \item
   \label{t3.iv}
   \emph{$\Gamma_{\mathsf{C}^{\infty}_0}(M,E)$ is dense in
     $\mathsf{D}^0_{\infty,\mathrm{loc}}(H_V)$ with respect to
     $\left\|\bullet\right\|_{H_V}$ and one has
     \begin{align}
       \mathsf{D}(H_V)
       =\left.\Big\{f\right|f,(\nabla^{\dagger}\nabla+V)f
          \in\Gamma_{\mathsf{L}^{2}}(M,E)\Big\}.
          \label{defi}
     \end{align}
    }

    \emph{Proof.}  Let $f \in
    \mathsf{D}^0_{\infty,\mathrm{loc}}(H_V)$.  By Lemma~\ref{hil} and
    the self-adjointness of $H_V$ we have $H_V\subset H^*_{V,\min}$,
    but it is well-known that (see for example p.644 in~\cite{Br})
    \begin{align}
      \mathsf{D}(H^*_{V,\min})=\left.\Big\{f\right|f,(\nabla^{\dagger}\nabla+V)f\in\Gamma_{\mathsf{L}^{2}}(M,E)\Big\}.\nn
    \end{align}
    In particular $\mathsf{D}^0_{\infty,\mathrm{loc}}(H_V)\subset
    \mathsf{D}(H^*_{V,\min}) $ implies $w:=\nabla^{\dagger}\nabla f+
    Vf\in\Gamma_{\mathsf{L}^2}(M,E)$. As $f$ is locally bounded with a
    compact support, one also has $Vf\in\Gamma_{\mathsf{L}^2}(M,E)$,
    so that $\nabla^{\dagger}\nabla f=w-
    Vf\in\Gamma_{\mathsf{L}^2}(M,E)$.  But now the assertion follows
    directly from Proposition~\ref{an1}, which is in fact a local
    result (and which again heavily uses that $f$ is locally bounded
    with a compact support).
    % and as a consequence
    % \begin{align}
    %   f\in \left.\Big\{h\right|h,Vh, \nabla^{\dagger}\nabla
    %   h\in\Gamma_{\mathsf{L}^{2}}(M,E)\Big\}.\label{sobo}
    % \end{align}

    Finally,~\eqref{defi} simply follows from the essential
    self-adjointness of $\tilde H_{V,\min}$, which follows
    from~\itemref{t3.ii} and the by now proven fact that
    $\Gamma_{\mathsf{C}^{\infty}_0}(M,E)$ is an operator core for
    $H_V$.\qedhere
  \end{myenumerate}
\end{proof}
%----------------------------------------------------------------------

We immediately get: 
%----------------------------------------------------------------------
\begin{Corollary}
  \label{ende}
  Theorem~\ref{dsk} holds, that is, under the assumptions of
  Theorem~\ref{T3}, the operator $\nabla^{\dagger}\nabla/2 +V$ is
  essentially self-adjoint on $\Gamma_{\mathsf{C}^{\infty}_0}(M,E)$,
  and its closure is semibounded from below.
\end{Corollary}
%----------------------------------------------------------------------
\begin{proof}
  Combining Theorem~\ref{T3} with Lemma~\ref{hil} immeadiately gives
  $H_{V,\min}=H^*_{V,\min}=H_V$.
\end{proof}
%----------------------------------------------------------------------

%----------------------------------------------------------------------
%
\section{Application to Hydrogen type problems on Riemannian
  3-manifolds}
\label{wasser}
%
%----------------------------------------------------------------------

In this section, we shall explain a typical application of
Theorem~\ref{dsk}: The essential self-adjointness of nonrelativistic
Hamiltonians corresponding to Hydrogen type atoms, with the
electron\rq{}s spin is taken into account. To this end, let us first
explain what the analogues of the Coulomb potential and the Pauli
operator are in a general curved setting. Here, we are going to
follow~\cite{G6} closely.
\begin{quote}
  \emph{Throughout Section~\ref{wasser}, we will assume that $M$ is a
    smooth connected Riemannian $3$-manifold without boundary.}
\end{quote}
Firstly, we want to point out that \lq\lq{}nonparabolicity\rq\rq{} is
the appropriate setting that admits natural analogues of the Coulomb
potential:
%----------------------------------------------------------------------
\begin{Definition}
  The Riemannian manifold $M$ is called \emph{nonparabolic}, if one has
  \[
  \int^{\infty}_0 p(t,x,y) \Id t <\infty\>\>\text{ for some (any) $x,y
    \in M$ with $x\ne y$.}
  \]
  Then 
  \[
  G \colon M\times M\longrightarrow (0,\infty],\>\>G(x,y):=
  \int^{\infty}_0 p(t,x,y) \Id t
  \]
  is called the \emph{Coulomb potential} on $M$.
\end{Definition}
%----------------------------------------------------------------------

It should be noted that nonparabolicity always implies noncompactness.
The essential point for the interpretation of $G$ as the Coulomb
potential is that $M$ is nonparabolic, if and only if $M$ admits a
positive Green\rq{}s function, and then $G$ is the minimal positive
Green\rq{}s function (see~\cite{G6} and the references therein for
these facts). The following criterion can be easily deduced from
Theorem~\ref{si}:

%----------------------------------------------------------------------
\begin{Lemma}
  \label{pa}
  Assume that there is a $C>0$ such that for all $t>0$ one has
  \begin{equation}
    \sup_{x\in M} p(t,x,x)\leq Ct^{-3/2}. \label{assy}
  \end{equation}
  Then $M$ is nonparabolic and there is a $\tilde{C}>0$ with
  \begin{equation}
    G(x,y)\leq \f{\tilde{C}}{\Id(x,y)} \text{ for all $x,y\in M$.}
    \label{marki0}
  \end{equation}
\end{Lemma}
%----------------------------------------------------------------------

Next, we will explain the natural analogues of the Pauli-operator in
our general setting. To this end, we give ourselves a
\emph{Pauli-Dirac structure} $(c,\nabla)$ on $M$ in the sense
of~\cite{G6}, that is, with a smooth Hermitian vector bundle $E \to M$
with $\mathrm{rank} E=2$,
\[
c\colon \mathrm{T}^*M\longrightarrow\End(E)
\]
is a Clifford multiplication\footnote{A Clifford multiplication $c$ is
  a morphism of smooth vector bundles such that for all $\alpha\in
  \Omega^1(M)$ one has
\begin{equation*}
c(\alpha)=-c(\alpha)^*,\>\>c(\alpha)^*c(\alpha)=\left|\alpha\right|^2. 
\end{equation*}
}and $\nabla$ is a Clifford connection\footnote{A Clifford connection is a Hermitian connection with the following property: for all $\alpha\in \Omega^1(M)$ and all
  $X\in
  \Gamma_{\mathrm{C}^{\infty}}(M,\mathrm{T}M),\psi\in\Gamma_{\mathrm{C}^{\infty}}(M,E)$
  one has
\[
\nabla_X(c(\alpha)\psi)=c(\nabla^{\mathrm{T} M}_X \alpha )\psi+
c(\alpha)\nabla_X \psi.
\]} with respect to $c$.

%----------------------------------------------------------------------
\begin{Remark}
  The existence of a Pauli-Dirac structure on $M$ is a topological
  restriction, namely, $M$ admits a Pauli-Dirac structure, if and only
  if $M$ is a $\mathrm{spin}^{\mathbb{C}}$ manifold. This fact has
  also been explained in~\cite{G6}.
\end{Remark}
%----------------------------------------------------------------------

The \emph{Pauli-Dirac operator} $\IDD(c,\nabla)$ with respect to
$(c,\nabla)$ is defined by
\[
\IDD(c,\nabla):=c\circ \nabla \colon
\Gamma_{\mathrm{C}^{\infty}_0}(M,E)\longrightarrow
\Gamma_{\mathrm{C}^{\infty}_0}(M,E),
\]
which is a linear first order differential operator with
$\IDD(c,\nabla)^{\dagger}=\IDD(c,\nabla)$. If $(e_j)$ is some smooth
local orthonormal frame for $\mathrm{T}M$, then one has
$\IDD(c,\nabla)=\sum_jc(e_j^*) \nabla_{e_j}$. Furthermore,
$\IDD(c,\nabla)^2$ is a generalized Laplacian on $M$ which is given by
the following Lichnerowicz formula:

%----------------------------------------------------------------------
\begin{Lemma}
  \label{lich}
  The differential form $\mathrm{tr} [\nabla^{2}]/
  \mathrm{i}\in\Omega^2(M)$ is real-valued and closed, and one has
  \begin{equation}
    \IDD(c,\nabla)^2= \nabla^{\dagger}\nabla +
\f{1}{4}\mathrm{scal}(\bullet)\mathbf{1} +  \f{1}{2}\sum_{i<j} \mathrm{tr}\left[
\nabla^{2}\right](e_i,e_j)c(e^*_i)c(e^*_j).\label{ggs}
  \end{equation}
\end{Lemma}
%----------------------------------------------------------------------

The last lemma makes it plausible (see also
Remark~\ref{phf}~\itemref{phf.b} below) to call $\IPP(c,\nabla):=
\IDD(c,\nabla)^2$ the \emph{Pauli-Dirac} operator with respect to
$(c,\nabla)$.

Clearly, if one has~\eqref{assy}, then $G(\bullet,y)$ exists and is
locally square integrable for any $y\in M$, and for any such $y$ and
$\kappa\geq 0$ one can consider the operator
\[
\tilde{H}(c,\nabla;\kappa,y):=\IPP(c,\nabla)-\kappa G(\bullet,y)\mathbf{1} 
\]
in $\Gamma_{\mathsf{L}^{2}}(M,E)$ with domain of definition
$\Gamma_{\mathrm{C}^{\infty}_0}(M,E)$, which gives rise to a symmetric
operator. Let us furthermore define the smooth potential
\begin{align}
  V(c,\nabla)
  :=  \f{1}{4}\mathrm{scal}(\bullet) \mathbf{1} 
       + \frac 12 \sum_{i<j} 
           \mathrm{tr}\left[ \nabla^2\right](e_i,e_j)
           c(e^*_i)c(e^*_j).
     \end{align}
With these preparations, Theorem~\ref{dsk} has the following important consequence:

%----------------------------------------------------------------------
\begin{Theorem}
  \label{hau}
  Assume that $M$ is geodesically complete with~\eqref{assy} and that $V(c,\nabla)$ admits a decomposition  
  \[
  V(c,\nabla)=V_1(c,\nabla)-V_2(c,\nabla)
  \]
  into potentials $V_j(c,\nabla)\geq 0$
  with $\left|V_2(c,\nabla)\right|\in\mathcal{K}(M)$. Then for any $\kappa\geq 0$ and $y\in M$, the operator
  $\tilde{H}(c,\nabla;\kappa,y)$ is essentially self-adjoint and its
  closure $H(c,\nabla;\kappa,y)$ is semibounded from below.
\end{Theorem}
%----------------------------------------------------------------------
\begin{proof}
  
Using ~\eqref{marki0} and Proposition~\ref{dhj} a), one
  easily checks that Theorem~\ref{dsk} can be applied with
  \[
    V :=\f{1}{4}\mathrm{scal}(\bullet)\mathbf{1} 
         +  \f{1}{2}\sum_{i<j} \mathrm{tr}
            \left[ \nabla^{2}\right](e_i,e_j)c(e^*_i)c(e^*_j)
             -\kappa G(\bullet,y)\mathbf{1},
  \]
  which proves the claim.
\end{proof}
%----------------------------------------------------------------------

%----------------------------------------------------------------------
\begin{Remark}
  \label{phf}
  \begin{myenumerate}{\alph}
  \item %(a)
    \label{phf.a} Let 
\begin{align}
  S(c,\nabla)
  :=& \int_M
  \Bigl| 
     \Bignorm{\frac 14 \mathrm{scal}(\bullet) \mathbf{1} 
       + \frac 12 \sum_{i<j} 
           \mathrm{tr}\left[ \nabla^2\right](e_i,e_j)
           c(e^*_i)c(e^*_j)
     }
  \Bigr|^2_x \mathrm{vol}(\Id x)\nn\\
&\in [0,\infty],\nn
\end{align}
where $\left|\left\|\bullet \right\|\right|_x$ stands for the
fiberwise Hilbert-Schmidt norm. Using $\left|\bullet \right|_x\leq
\left|\left\|\bullet \right\|\right|_x$ and Proposition~\ref{dhj} a),
one sees that the assumption on $V(c,\nabla)$ in Theorem~\ref{hau} is
obviously satisfied under~\eqref{assy}, if $S(c,\nabla)<\infty$. This
variant of Theorem~\ref{hau} has been deduced in~\cite{G6} with
completely different methods, namely, using results of~\cite{Br}
(which rely on pure PDE methods).

  \item %(b)
    \label{phf.b}
    In the situation of Theorem~\ref{hau}, the operator
    $H(c,\nabla;\kappa,y)$ can be interpreted~\cite{G6} as the
    nonrelativistic Hamiltonian corresponding to an atom with one
    electron and a nucleus with $\sim \kappa$ protons, where the
    electron\rq{}s spin has been taken into account and the nucleus is
    considered as fixed in $y$ with respect to the electron.  Here, in
    view of Lemma~\ref{lich}, the underlying magnetic field is given
    by $\mathrm{tr} [\nabla^{2}]/ \mathrm{i}\in\Omega^2(M)$. In
    particular, the above mentioned assumption $S(c,\nabla)<\infty$ is reasonable from
    the physics point of view, for it corresponds in a certain sense
    to a \lq\lq{}finite magnetic self-energy\rq\rq{} (it is essential
    for this interpretation to take the Hilbert-Schmidt norm in the
    definition of $S(c,\nabla)$).
  \end{myenumerate}
\end{Remark}
%----------------------------------------------------------------------

%----------------------------------------------------------------------
\subsection*{Acknowledgements}

The first author (BG) is indebted to Ognjen Milatovic for many
discussions on essential self-adjointness in the past three years, in
particular, for bringing the reference~\cite{grumm} into our attention
(which helped us to remove an unnecessary assumption from the original
version of Theorem~\ref{dsk}). Both authors kindly acknowledge the
financial support given by the SFB~647 ``Space---Time---Matter'' at
the Humboldt University Berlin, where this work has been started.

% ----------------------------------------------------------------------
\appendix
\section{Friedrichs mollifiers}
%
%----------------------------------------------------------------------

We record the following result on Friedrichs mollifiers here. Let $M$
be a smooth connected Riemannian manifold without boundary, $E\to M$ a
smooth Hermitian vector bundle, $\nabla$ a Hermitian covariant
derivative in $E$, and $V \colon M\to\End(E)$ a potential.

%----------------------------------------------------------------------
\begin{Proposition}
  \label{an1}
  Let $\left|V\right|\in\mathsf{L}^2_{\mathrm{loc}}(M)$ and assume
  that $f\in\Gamma_{ \mathsf{L}^{\infty}_{\mathrm{loc}}}(M,E)$ is
  compactly supported with $\nabla^{\dagger}\nabla f\in \Gamma_{
    \mathsf{L}^{2}_{\mathrm{loc}}}(M,E)$ in the sense of
  distributions. Then there is a sequence $(f_n)_{n\in\IN}\subset
  \Gamma_{\mathsf{C}^{\infty}_0}(M,E)$ such that
  \begin{align}
    &\lim_{n\to\infty}\left\|f_n-f\right\|=0,\nn\\
    &\lim_{n\to\infty}\left\|\nabla^{\dagger}\nabla f_n -\nabla^{\dagger}\nabla
f\right\|=0,\nn\\
    &\lim_{n\to\infty}\left\|V f_n -Vf\right\|=0.\nn
  \end{align} 
\end{Proposition}
%----------------------------------------------------------------------

%----------------------------------------------------------------------
\begin{Remark}
  Note that one indeed has $f\in\Gamma_{ \mathsf{L}^{2}}(M,E)$, which
  follows from $f\in\Gamma_{ \mathsf{L}^{\infty}_{\mathrm{loc}}}(M,E)$
  and the fact that $f$ has a compact support. Furthermore,
  $\nabla^{\dagger}\nabla f\in \Gamma_{ \mathsf{L}^{2}}(M,E)$ follows
  from $\nabla^{\dagger}\nabla f\in \Gamma_{
    \mathsf{L}^{2}_{\mathrm{loc}}}(M,E)$ and the fact that
  $\nabla^{\dagger}\nabla f$ has a compact support.
\end{Remark}
%----------------------------------------------------------------------

%----------------------------------------------------------------------
\begin{proof}[Proof of Proposition~\ref{an1}]
  Since most of the arguments should be well-known, we only sketch the
  proof. Let $m:=\dim M$ and let $d$ be the fiber dimension of $E$.
  Since $f$ is compactly supported, we can use a partition of unity
  argument to assume that $f$ is supported in a relatively compact coordinate domain
  $U\subset M$ (which is identified with an open subset of $\IR^m$) such that
  there is a smooth orthonormal frame for $E$ over $U$, and we denote
  the components of $f$ in this frame with $f^{(1)},\dots,f^{(d)}$.
  Now take some $0\leq j_r\in\mathsf{C}^{\infty}_0(\IR^m)$ with
  $j(z)=0$ for $|z|\geq 1$ and
  \[
  \int_{\IR^m}j(z)\Id z=1.
  \]
  For $r>0$ let $j_r\in\mathsf{C}^{\infty}_0(\IR^m)$ be given by
  $j_r(z)=r^{-m}j(r^{-1} z)$. Let $r>0$ be small enough in the
  following such that the functions
  \begin{align}
    x\longmapsto  \int_{\IR^m} j_r(x-y)f^{(i)}(y) \Id y,\>\>\>i=1,\dots
    d,
    \label{des}
  \end{align}
  define an element
  \[
  f_r\in\Gamma_{\mathsf{C}^{\infty}_0}(U,E)\subset
  \Gamma_{\mathsf{C}^{\infty}_0}(M,E).
  \]
  Since the sections $f_r-f$ and $\nabla^{\dagger}\nabla f_r
  -\nabla^{\dagger}\nabla f$ are compactly supported, the convergence
  \begin{align}
    \lim_{r\to 0+}\left\|f_r-f\right\|=0\label{coni}
  \end{align}
  follows from Lemma~5.13~(ii) in~\cite{Br}, and
  \[
  \lim_{r\to 0+}\left\|\nabla^{\dagger}\nabla f_r
    -\nabla^{\dagger}\nabla f\right\|=0
  \]
  follows from the $\mathsf{L}^2_{\mathrm{loc}}$-version of
  Proposition~5.14 in~\cite{Br}, which can be proven with analogous
  arguments. Note that so far we have only used that $f$ is locally
  square integrable with a compact support.

  The local boundedness assumption on $f$ comes into play as follows:
  Namely, this assumption combined with the compact support assumption
  implies that $f$ is actually \emph{bounded} and so~\eqref{des}
  implies
  \begin{align}
    \left|f_r(x)\right|_x\leq \left\|f\right\|_{\infty}\text{ for
      all $x$, $r$}.\label{des2}
  \end{align}
  Since (in view of~\eqref{coni}) we may assume that $f_r\to f$ a.e.
  in $M$, and since $f_r$ has a compact support, the required
  convergence
  \[
  \lim_{r\to 0+}\left\|V f_r -V f\right\|=0
  \]
  now follows from~\eqref{des2} and dominated convergence.
\end{proof}
%----------------------------------------------------------------------

%----------------------------------------------------------------------
%
\section{Finite speed of propagation }
%
%----------------------------------------------------------------------

The following lemma is usually referred to as Chernoff\rq{}s finite
speed of propagation method~\cite{chernoff}.  Let $M$ be a smooth
connected Riemannian manifold without boundary, and let $E\to M$ be a
smooth Hermitian vector bundle.

%----------------------------------------------------------------------
\begin{Lemma}
  \label{chernoff1}
  Let $S$ be a self-adjoint nonnegative operator in
  $\Gamma_{\mathsf{L}^{2}}(M,E)$. Assume furthermore that
  $\mathsf{D}^0(S)$, the compactly supported elements of
  $\mathsf{D}(S)$, are dense in $\Gamma_{\mathsf{L}^{2}}(M,E)$ and
  that for any $f\in \mathsf{D}^0(S)$ and any $t>0$, the section
  $\cos(t\sqrt{S})f$ has a compact support. Then $\mathsf{D}^0(S)$ is
  an operator core for $S$.
\end{Lemma}
%----------------------------------------------------------------------
\begin{proof} 
  The proof is a straightforward generalisation of the proof of
  Theorem~3 in~\cite{grumm}.
\end{proof}
%----------------------------------------------------------------------

%----------------------------------------------------------------------
%
\section{Path ordered exponentials}
%
%----------------------------------------------------------------------

In the following lemma, we collect some known facts about path ordered
exponentials for the convenience of the reader:

%----------------------------------------------------------------------
\begin{Lemma}
  \label{poe} Let $\IHH$ be a finite dimensional Hilbert space, let
  $T\in (0,\infty]$ and let $F\in
  \mathsf{L}^1_{\mathrm{loc}}([0,T),\ILL(\IHH))$. Then the following
  assertions hold:
  \begin{myenumerate}{\alph}
  \item %(a)
    There is a unique weak ($=\mathsf{AC}_{\mathrm{loc}}$) solution
    $Y:[0,T)\to \ILL(\IHH)$ of the ordinary initial value problem
    \begin{equation}
      \f{\Id}{\Id t} Y(t)=Y(t)F(t),\>\>Y(0)=\mathbf{1}.
    \end{equation}

  \item %(b)
    For any $0\leq t<T$ one has
    \begin{equation}
      Y(t)=\mathbf{1}+\sum^{\infty}_{k=1}\int_{0\leq s_1\leq \dots\leq
        s_k\leq t} F(s_1)\dots F(s_k) \Id s_1\dots \Id
      s_k.\label{pato}
    \end{equation}

  \item %(c)
    If $F(\bullet)$ is Hermitian a.e. in $[0,T)$ and if there exists a
    real-valued function $c\in\mathsf{L}^1_{\mathrm{loc}}[0,T)$ such
    that for all $v\in\IHH$ it holds that
    \[
    \left\langle F(\bullet) v, v\right\rangle_{\IHH}\leq c(\bullet)
    \left\|v\right\|^2_{\IHH}\>\>\text{ a.e. in $[0,T)$,}
    \]
    then one has
    \[
    \left\|Y(t)\right\|_{\IHH}\leq \e^{\int^t_0 c(s) \Id
      s}\>\>\text{ for all $0\leq t<T$.}
    \]
  \end{myenumerate}
\end{Lemma}
%----------------------------------------------------------------------

%----------------------------------------------------------------------

%----------------------------------------------------------------------
\begin{proof} 
  See~\cite{dollard} and the Appendix~C of~\cite{G2}.
\end{proof}
%----------------------------------------------------------------------

%\begin{Lemma}\label{chernoff2} Let $V$ be smooth and nonnegative, let
%$N_1,N_2\subset M$ be open subsets and let $f_1,f_2\subset 
%\Gamma_{\mathsf{L}^{2}}(M,E)$ be such that $\mathrm{supp}(f_j)\subset N_j$ for
%$j=1,2$. Then one has   
%\begin{align}
%\iprod{\cos(t\sqrt{H_V})f_1}{f_2}=0\text{ for all $0<t <\Id(N_1,N_2)$}.
%\end{align}
%\end{Lemma}

%----------------------------------------------------------------------

\end{document}